\documentclass[graybox]{svmult}

\usepackage{mathptmx}       % selects Times Roman as basic font
\usepackage{helvet}         % selects Helvetica as sans-serif font
\usepackage{courier}        % selects Courier as typewriter font
\usepackage{type1cm}        % activate if the above 3 fonts are
\usepackage{makeidx}         % allows index generation
\usepackage{graphicx}        % standard LaTeX graphics tool
\usepackage{multicol}        % used for the two-column index
\usepackage[bottom]{footmisc}% places footnotes at page bottom

\usepackage{bm}% bold math
\usepackage{amsmath}
\usepackage{amssymb}
\usepackage{mathrsfs}
\usepackage{color}

\DeclareSymbolFont{newfont}{OML}{cmm}{m}{it}% Computer Modern math font
\DeclareMathSymbol{\varrho}{3}{newfont}{37}% Symbol 37

\def \rmd{\mathrm{d}}

\def \Ad{\mathrm{Ad}}
\def \vol{\mathrm{vol}}

\makeindex             % used for the subject index
\begin{document}

% \title*{Topological constraints and structures in macro (fluid and plasma) systems}
\title*{Epi-two-dimensional flow and generalized enstrophy}
\titlerunning{Epi-2D flow and generalized enstrophy}
\author{Zensho Yoshida and Philip J. Morrison}
\institute{Zensho Yoshida \at Department of Advanced Energy, University of Tokyo, Chiba 277-8561, Japan, \email{yoshida@ppl.k.u-tokyo.ac.jp}
\and Philip J. Morrison \at Department of Physics and Institute for Fusion Studies, University of Texas at Austin, TX 78712-1060, USA, \email{morrison@physics.utexas.edu}}
\maketitle

\abstract{
The conservation of the enstrophy ($L^2$ norm of the vorticity $\omega$)
plays an essential role in the physics and mathematics of two-dimensional (2D) Euler fluids.
Generalizing to compressible ideal (inviscid and barotropic) fluids, the generalized enstrophy
$\int_{\Sigma(t)} f(\omega/\rho)\rho\rmd^2 x$\, 
($f$ an arbitrary smooth function, $\rho$ the density, and $\Sigma(t)$ an arbitrary 2D domain co-moving with the fluid) 
is a constant of motion, and plays the same role.
On the other hand, for the  three-dimensional (3D) ideal fluid,
the helicity $\int_{M} \bm{V}\cdot\bm{\omega}\,\rmd^3x$\, 
($\bm{V}$  the flow velocity, $\bm{\omega}=\nabla\times\bm{V}$, and ${M}$  the three-dimensional domain containing the fluid) is conserved.
Evidently, the helicity degenerates in a 2D system, and the (generalized) enstrophy emerges as a compensating constant.
This transition of the constants of motion is a  reflection of an  essential difference between 2D and 3D systems,
because the conservation of the (generalized) enstrophy imposes stronger constraints, than the helicity, on the flow.
In this paper, we make a deeper inquiry into the helicity-enstrophy interplay: 
the ideal fluid mechanics is cast into a Hamiltonian form in the phase space of Clebsch parameters, 
generalizing 2D to a wider category of epi-2D flows  (2D  embedded in 3D has zero helicity, while the converse is not true -- our  epi-2D category encompasses a wider class of zero-helicity flows); how  helicity degenerates  and  is substituted by a new constant is delineated; and  how a further generalized enstrophy is introduced as a constant of motion applying to epi-2D flow is described. 
}

\section{Introduction}
The aim of this paper is to elucidate,
from  the perspective of Hamiltonian dynamics\,\cite{Morrison-RMP},
how two-dimensional (2D) flow is different from general three-dimensional (3D) flow.
Phenomenologically, 2D flow is often very different from 3D flow in that the former is less-turbulent
and is more capable of generating  and sustaining  large-scale vortical structures
-- typhoons, jet streams, polar vortexes being  spectacular examples of such structures created in atmospheric 2D flow.
If we could delineate the root cause of such special behavior in  2D, we might  be able to obtain a flow `intermediate' between 2D and 3D, where the `regularity' of 2D flow is maintained.  As we will show, such is indeed possible. 

We invoke the \emph{helicity} as the key parameter for characterizing the transition from 3D to 2D (see Sec.\,\ref{sec:preliminaries}).
As is well known, the helicity is a constant of motion in an `ideal' flow
(in this paper  {ideal} means inviscid and barotropic).
In 2D geometry, however, the helicity degenerates to zero; but 
as a compensation, the \emph{enstrophy} (or its generalization, cf.\  Remark\,\ref{remark:enstrophy})
becomes a nontrivial constant (see Secs.\,\ref{subsec:2D} and \ref{subsec:enstrophy}).
The conservation of the (generalized) enstrophy is a  most essential property for 
distinguishing 2D from 3D.
The enstrophy is a higher-order functional in comparison with the helicity,
and its conservation is deemed to be reason for the aforementioned difference between 2D and 3D systems.
Even when the constancy of the enstrophy or the helicity is broken by the inclusion of  dissipation,
the macroscopic structure of the fluid system is strongly influenced by these ideal constants of motion (cf.\,\cite{Hasegawa1985}).

Needless to say, zero-helicity flow is not necessarily 2D.
In Sec.\,\ref{sec:epi-2D}, we introduce our category of `epi-2D' flow that maintains the basic properties
of zero-helicity and enstrophy conservation, while not necessarily being 2D.  Having cast  ideal fluid mechanics into a Hamiltonian form  in the phase space of Clebsch parameters (e.g.\ \cite{Clebsch,Zakharov,Marsden1983,Jackiw,Yoshida2009}), the
category of epi-2D flow is, then, defined as a reduction of the phase space.
One of the reduced parameter  used is  a  \emph{phantom}\,\cite{YoshidaMorrisonFDR2014,YoshidaMorrison2016},
by which we define a generalized enstrophy.  
% Although a full treatment of meaning of  epi-2D flow on the Lagrangian variable level is beyond the present scope, 
In Sec.\,\ref{sec:particle} we introduce the notion of an epi-2D particle to elucidate our theory.

%%%%%%%%%%%%%%%%%%%%%%%%%%%%%%%%%%%%%
%%%%%%%%%%%%%%%%%%%%%%%%%%%%%%%%%%%%%
%%%%%%%%%%%%%%%%%%%%%%
\section{Preliminaries}
\label{sec:preliminaries}

%%%%%%%%%%%%%%%%%%%%%%
\subsection{Three-dimensional fluid mechanics}
\label{subsec:3D}

We start by reviewing the basic equations of fluid mechanics and the associated  conservation laws.
Here we use the conventional notation of 3D vector analysis, with vector fields  denoted by bold-face symbols.

Let ${M}$ be a 3D domain containing an ideal fluid.
For simplicity, we assume ${M}=T^3$, the 3-torus, and ignore the effect of boundaries.
We denote by ${\rho}$ the mass density, $\bm{V}$ the fluid velocity,
and $P$ the pressure.
Thus, the  governing equations are 
\begin{eqnarray}
& & \partial_t{\rho} = -\nabla\cdot(\bm{V} {\rho}),
\label{mass_conservation}
\\
& & \partial_t \bm{V}=-(\bm{V}\cdot\nabla)\bm{V} - \rho^{-1}\nabla P.
\label{momentum}
\end{eqnarray}
Assuming  a barotropic relation $P=P(\rho)$, with  $ \rho^{-1}\nabla P = \nabla h$ for  $h=h(\rho)$, 
the energy of the system is
\begin{equation}
H = \int_{M} \left[ \frac{1}{2} |\bm{V}|^2 + \varepsilon(\rho)  \right]\,\rho\,\rmd^3 x,
\label{energy}
\end{equation}
where $\varepsilon(\rho)$ is the specific thermal energy, satisfying $\partial(\rho\epsilon(\rho))/\partial\rho = h(\rho)$.  
Evidently, $\rmd H/\rmd t=0$.

The vorticity $\bm{\omega}=\nabla\times\bm{V}$ obeys the vorticity equation, obtained by 
taking the curl of (\ref{momentum}), i.e., 
\begin{equation}
\partial_t \bm{\omega}= \nabla\times(\bm{V}\times\bm{\omega}).
\label{vorticity-3D}
\end{equation}
Evidently, the total mass $N=\int_{M} \rho\,\rmd^3x$ is a constant of motion, along with 
another conserved quantity,  the helicity: 
\begin{equation}
C = \int_{M} \bm{V}\cdot\bm{\omega}\,\rmd^3 x.
\label{Helicity_naive}
\end{equation}
Using (\ref{momentum}) and (\ref{vorticity-3D}), we easily verify $\rmd C/\rmd t =0$.\footnote{
In this work, we do not argue for the the existence  of regular solutions of the model equations.
The conservation laws discussed are, therefore,  \emph{a priori} relations satisfied by all regular solutions if they exist. 
}

%%%%%%%%%%%%%%%%%%%%%%%
\subsection{Two-dimensional fluid mechanics}
\label{subsec:2D}

To compare 2D and 3D systems, it is convenient to immerse a 2D system into 3D space.
For simplicity, we consider a flat torus $T^2$, on which we define Cartesian coordinates $x$ and $y$.
We add a `perpendicular' coordinate $z$ and extend $T^2$ to $T^3$, with 
 $\bm{e}_z=\nabla z$, which we call the perpendicular vector.
Now we may define a 2D system by the reduction of the 3D system with $\bm{e}_z\cdot ~=0$ and $\partial_z=0$.
Indeed, a 2D fluid model is formulated by such a reduction.\footnote{
We note that $\partial_z=0$ does not mean that the system extends in the $z$-direction  homogeneously;
instead, we consider a thin system in which  variation  of physical quantities in the $z$-direction 
is much larger than in the $x$ and $y$ directions.  Thus,   $\partial_z$ can be separated from
$\partial_x$ and $\partial_y$. 
}
We interpret a 2D flow $\bm{v} = ( v_x, v_y)^{\mathrm{T}}$ as a special 3D flow such that $\bm{V}=( v_x, v_y, 0)^{\mathrm{T}}$.
The vorticity can be defined as $\bm{\omega}=\nabla\times\bm{V}=\omega \bm{e}_z$
with $\omega=\partial_x v_y - \partial_y v_x$. In which case the  vorticity equation (\ref{vorticity-3D}) reduces to  a single-component equation:
\begin{equation}
\partial_t {\omega}= -\nabla\cdot (\bm{v}\omega) .
% - \bm{v}\cdot\nabla\omega - \omega \nabla\cdot\bm{v} .
\label{vorticity-2D}
\end{equation}

Because $\bm{V}\cdot\bm{\omega} = 0$, the helicity conservation is now trivial, $C\equiv0$; however, 
interestingly, a new constant emerges that  replaces the degenerated helicity.\footnote{
Fukumoto\,\cite{Fukumoto2008} points out that the helicity and the generalized enstrophy can be
unified by the concept of \emph{cross helicity}.
}
By (\ref{vorticity-2D}) and the mass conservation law,
which now reads $\partial_t\rho = -\nabla\cdot(\bm{v} \rho)$,
we obtain, the following equation for the \emph{potential vorticity}, $\vartheta=\omega/\rho$: 
\begin{equation}
\partial_t \vartheta = -\bm{v}\cdot\nabla \vartheta , 
\label{potential-vorticity}
\end{equation}
and we define the \emph{generalized enstrophy}  by 
\begin{equation}
{Q} = \int_{M} f(\vartheta) \rho\,\rmd^2 x,
\label{G-enstrophy}
\end{equation}
where $f$ is an arbitrary $C^1$-class function.
Using (\ref{potential-vorticity}) and the mass conservation law, we can easily verify that $\rmd {Q}/\rmd t=0$.

\begin{remark}
\label{remark:enstrophy}
For an incompressible flow ($\nabla\cdot\bm{v}=0$), we may assume $\rho=$ constant, and then,
(\ref{G-enstrophy}) has a special form of $\int_{M} \omega^2 \rmd^2x $,
which is the usual \emph{enstrophy}. 
\end{remark}

We end this introductory section with drawing attention to the fact that all constants of motion,
i.e. the total mass $N$, the helicity $C$, and the generalized enstrophy ${Q}$ are defined by
the spatial integrals over the 3D or 2D domain.
This means that the integrand of a constant of motion defines an $n$-form ($n$  the spatial dimension)
in the language of differential geometry.
In the following analysis, this fact guides our formulation of  generalized enstrophy.

%%%%%%%%%%%%%%%%%%%%%%%%%%%%%%%%%%%%%
%%%%%%%%%%%%%%%%%%%%%%%%%%%%%%%%%%%%%
%%%%%%%%%%%%%%%%%%%%%%
\section{Topological invariants in fluid motion}
\label{sec:helicity}

\subsection{Hamiltonian formalism of ideal fluid motion}
\label{subsec:ideal_flow}

For the study of geometrical properties of fluid mechanics, we   reformulate the
governing equations in the framework of differential geometry.
We first introduce a \emph{phase space} $X$ that hosts the underlying state vectors $\bm{\xi}$;
the physical quantity $\bm{u} = (\rho, \bm{V})^{\textrm{T}}$ 
($\in\mathcal{V}$,  the space of physical variables) is some function parameterized by $\bm{\xi}$.
Let 
\begin{equation}
\bm{\xi}  =  (\varphi, \varrho, q, p, r, s)^{\mathrm{T}} ~\in X,
\label{canonical_variables}
\end{equation}
where $\xi_1=\varphi, \xi_3=q, \xi_5 =r$ are 0-forms and $\xi_2=\varrho, \xi_4=p, \xi_6=s$ are $n$-forms in the base space ${M}=T^3$.
We assume $\xi_j$ ($j=1,\cdots,6$) are smooth (i.e. $C^\infty$-class) functions.
The dual space $X^*$ is the Hodge-dual of $X$, i.e. the odd number components of $\bm{\eta}\in X^*$ are $n$-forms and the even number components are 0-forms.
The pairing of $X^*$ and $X$ is
\footnote{
Here the phase space (function space)  $X$
may be viewed as a cotangent bundle of $X_q =\{ (\xi_1, \xi_3, \xi_5)^{\mathrm{T}} ;\, \xi_j \in \mathrm{\Omega}^0(M)\}$. 
For $F\in C^\infty(X)$, $\partial_{\bm{\xi}} F\in X^*$ (to be defined in (\ref{gradient}))
may be regarded as a `1-form' on $X$.
Hence, the duality of $X^*$ and $X$ corresponds to that of `co-vectors' and `vectors'.
At the same time, the components of the field $\bm{\xi}\in X$ are differential forms (0-forms and n-forms) on the `base space' $M$;
the `Hodge-duality' of $X^*$ and $X$ is in the sense of the differential forms on $M$,
while the duality (\ref{pairing}) is in the sense of `co-vectors' and `vectors' on the function space $X$.
}
\begin{equation}
\langle \bm{\eta}, \bm{\xi} \rangle = \sum_{j} \int_{M} \eta_j \wedge \xi_j ,
\quad \bm{\eta}\in X^*, ~ \bm{\xi}\in X.
\label{pairing}
\end{equation}
On the space $C^\infty(X)$ of \emph{observables}, we define a canonical Poisson bracket
\begin{equation}
\{ F, G \} = \langle \partial_{\bm{\xi}} F, J \partial_{\bm{\xi}} G \rangle,
\label{canonical_Poisson_bracket}
\end{equation}
where $F, G \in C^\infty(X)$, $\partial_{\bm{\xi}} F$ is the gradient of $F$ defined by
\begin{equation}
F(\bm{\xi}+\epsilon \bm{\zeta}) - F(\bm{\xi}) = \epsilon \langle \partial_{\bm{\xi}} F, \bm{\zeta} \rangle + O(\epsilon^2)
\quad (\forall \bm{\zeta}\in X),
\label{gradient}
\end{equation}
and $J:\, X^* \rightarrow X$ is the symplectic operator 
\begin{equation}
J = J_c\oplus J_c\oplus J_c,
\quad 
J_c = \left( \begin{array}{cc}
0 & I \\
-I & 0
\end{array}\right).
\label{symplectic}
\end{equation}
We denote by $C^\infty_{\{~,~\}}(X)$ the Poisson algebra of observables on $X$.
The adjoint representation of  Hamiltonian dynamics is,
for a given Hamiltonian $H$,
\begin{equation}
\frac{\rmd}{\rmd t} F = \{ F, H \},
\label{adjoint}
\end{equation}
which is equivalent to Hamilton's equation of motion
\begin{equation}
\frac{\rmd}{\rmd t} \bm{\xi}= J \partial_{\bm{\xi}} H .
\label{Hamilton_eq}
\end{equation}

We relate the physical quantity $\bm{u}\in \mathcal{V}$ and $\bm{\xi}\in X$ by
$\rho\Leftrightarrow \varrho^*$
(i.e. $\varrho^* \vol^n = \varrho$ with the volume $n$-form $\vol^n$; here $n=3$),\footnote{Here we denote by $\alpha^*$ the Hodge dual of a differential form $\alpha$.}
and  
\begin{equation}
\bm{V} \Leftrightarrow    \wp = \rmd\varphi + \check{p} \rmd q + \check{s} \rmd r ,
\quad \left(\check{p}={p^*}/{\varrho^*}, ~\check{s}={s^*}/{\varrho^*} \right).
\label{Clebsch}
\end{equation}
Writing a vector as (\ref{Clebsch}) is called the  \emph{Clebsch parameterization}.
The five Clebsch parameters $(\varphi, q, \check{p}, r, \check{s})$ are sufficient to represent every 3-vector (1-form in 3D space)\,\cite{Yoshida2009}.
Inserting (\ref{Clebsch}) into the fluid energy (\ref{energy}), we obtain a Hamiltonian
\begin{equation}
H(\bm{\xi}) = \int_{M} \left[ \frac12 {\big|\rmd\varphi + (p^*/\varrho^*) \rmd q + (s^*/\varrho^*) \rmd r\big|^2} + \varepsilon(\varrho^*)  \right]\,\varrho.
\label{Hamiltonian}
\end{equation}
With this $H$, the equation of motion (\ref{Hamilton_eq}) reads
% (denoting by $L_{\bm{V}}$ the Lie derivative along the vector $\bm{V}\in T{M}$)
\begin{eqnarray}
& &\widetilde{\mathcal{L}}_{\bm{V}} \varphi = h-{V^2}/{2},
\label{H-1} \\
& & \widetilde{\mathcal{L}}_{\bm{V}} q = 0, \quad\widetilde{\mathcal{L}}_{\bm{V}} r = 0,
\label{H-2} \\
& &\widetilde{\mathcal{L}}_{\bm{V}} \varrho = 0, \quad\widetilde{\mathcal{L}}_{\bm{V}} p = 0, \quad \widetilde{\mathcal{L}}_{\bm{V}} s = 0 , 
\label{H-3} 
\end{eqnarray}
where we denote
\begin{equation}
\widetilde{\mathcal{L}}_{\bm{V}} = \partial_t + \mathcal{L}_{\bm{V}},
\label{extended_Lie_detivative}
\end{equation}
with $\mathcal{L}_{\bm{V}}$ being  the conventional Lie derivative along the vector $\bm{V}\in T{M}$.\footnote{
Here $\bm{V}$ is regarded as a vector $\in TM$ through the following identification.
By (\ref{Clebsch}), $\bm{V} \Leftrightarrow \wp \in T^*M$.
In the Hamiltonian (\ref{Hamiltonian}), $|\wp|^2 \Leftrightarrow \bm{V}^\dagger\cdot\bm{V}$ with the dual
$\bm{V}^\dagger=\bm{V}\in TM$.
% We point out that a similar Lie derivative construction exists for extended MHD \cite{LMM16} in its natural (physical) variables.
}
\footnote{
When considering  a relativistic fluid, we generate a diffeomorphism group $\mathrm{e}^{\tau U}$
($\tau$ the  proper time),
acting on a space-time manifold $\tilde{M}=\mathbb{R}\times{M}$,
by a space-time velocity $U\in T\tilde{M}$.
Then, the space-time derivative $\widetilde{\mathcal{L}}_{\bm{V}}$ is replaced by the
natural Lie derivative $\mathcal{L}_U$.
When  $\mathcal{L}_U$ applies to a differential form $\alpha$ on $\tilde{M}$,
the temporal and spatial components are mixed up (cf.\,\cite{YKY2014}). }

The first equation of (\ref{H-3}) is nothing but the mass conservation law (\ref{mass_conservation}).
Evaluating $\partial_t\bm{V}$ by inserting (\ref{Clebsch}) and using (\ref{H-1})-(\ref{H-3}), we
obtain (\ref{momentum}).
Hence, Hamilton's equation (\ref{Hamilton_eq}) with the Hamiltonian (\ref{Hamiltonian})
describes the fluid motion obeying (\ref{mass_conservation}) and (\ref{momentum}).\footnote{
See \cite{Lin} for the underlying action principle that yields the canonical system of Hamilton's
equation (\ref{H-1})--(\ref{H-3}).
}

%%%%%%%%%%%%%%%%%%%%%
\subsection{Gauge symmetry and helicity}
\label{subsec:gauge}

In (\ref{Clebsch}), the Clebsch parameters are apparently a redundant representation of  a 3-vector $\bm{V}$.
In fact, the map $X \rightarrow \mathcal{V}$ is not an injection (although a surjection)\,\cite{Yoshida2009}.
For example, the transformation
\begin{equation}
\varphi \mapsto \varphi + \epsilon
\quad (\epsilon\in\mathbb{R})
\label{gauge-0}
\end{equation}
does not change the physical quantity $\bm{u}\in\mathcal{V}$.
Such a map is called a \emph{gauge transformation}.
We find that the map (\ref{gauge-0}) is a Hamiltonian flow generated by the constant of motion $N=\int_{M} \varrho$,
i.e. the map (\ref{gauge-0}) may be written as $(I+\epsilon J\partial_{\bm{\xi}}N)$.
Or, the co-adjoint orbit $\Ad_N^*(\epsilon)$ is a gauge-transformation group of the Clebsch parameterization.

The helicity $C$, which now reads
\begin{equation}
C=\int_{M} \wp\wedge\rmd\wp,
\label{helicity-2}
\end{equation}
yields a different gauge group $\Ad_C^*(\epsilon)$ (see\,\cite{TanehashiYoshida2015} for the explicit form the corresponding gauge transformation).

\begin{remark}
\label{remark:Casimir}
If we denote by $\{ F, G \}_{\mathcal{V}}$ the Poisson bracket of (\ref{canonical_Poisson_bracket})
evaluated only for observables $F$ and $G$ in which $\bm{\xi}$ appears in terms of the Clebsch parameterized $\bm{u}\in\mathcal{V}$,
we obtain
\begin{equation}
\{ G, N \}_{\mathcal{V}} =0, \quad
\{ G, C \}_{\mathcal{V}} =0
\quad \forall G.
\label{Casimir}
\end{equation}
Hence, $N$ and $C$ are the Casimir elements of the reduced Poisson algebra $C^\infty_{\{~,~\}_{\mathcal{V}}}(\mathcal{V})$.
See e.g.\,\cite{Marsden} for general discussion on \emph{reduction} of Poisson brackets.
The existence of Casimir elements is the characteristics of \emph{noncanonical Poisson brackets}\,\cite{Morrison-RMP,MG80}
We note that $C^\infty_{\{~,~\}_{\mathcal{V}}}(\mathcal{V})$ has much more (in fact, infinitely many) 
topological constraints (constants of motions) that are not integrable, i.e., do not define Casimir elements (see \cite{YoshidaMorrisonFDR2014}).
\end{remark}

\subsection{Two-dimensional system and generalized enstrophy}
\label{subsec:enstrophy}

In a 2D system (${M} = T^2$), we can parameterize a general 2D velocity as
\begin{equation}
\bm{V} \Leftrightarrow  \wp = \rmd\varphi + \check{p} \rmd q  .
\label{Clebsch-2D}
\end{equation}
Now only three Clebsch parameters $\varphi$, $\check{p}$, and $q$ suffice\,\cite{Yoshida2009}.
Hence, the phase space is 
\begin{equation}
Z = \{ \bm{\zeta}= (\varphi,\varrho, q, p)^{\mathrm{T}};
~\varphi, q \in \mathrm{\Omega}^0(T^2),~\varrho, p\in \mathrm{\Omega}^2(T^2) \}.
\label{2D_phase_space}
\end{equation}
All other formalisms are the same as the case of 3D systems.
% The vorticity $\omega = \rmd \wp$ ($\wp = \rmd\varphi + ({p}^*/\varrho^*) \rmd q $)
However, because the 3-form $\wp\wedge\rmd\wp$ cannot be defined in 2D space 
we do not have the helicity conservation law.

As mentioned in Sec.\,\ref{subsec:2D},   a different constant of motion emerges in 2D,
the generalized enstrophy, which is a spatial (2D) integral of a 2-form
that involves  the vorticity $\omega=\rmd \wp$.
As a preparation for the development in the next section, we
re-formulate the generalized enstrophy in the language of differential geometry
(with a slight extension), and  re-prove its conservation.
% We denote $\omega = \omega^* \vol^2$ and $\varrho = \varrho^* \vol^2$.
% The scalars $\omega^*$ and $\rho^*$ are, respectively, the Hodge duals of the 2-forms
% $\omega$ and $\varrho$.
Let 
\begin{equation}
{Q} = \int_{\Sigma(t)} f(\omega^*/\varrho^*) \varrho ,
\label{enstrophy-1}
\end{equation}  
where $f$ is an arbitrary smooth function, and $\Sigma(t)\subset{M}$ is a co-moving `volume' (in fact, a 2D surface).
Notice that the integral is evaluated on a subset $\Sigma(t)$ that is moved by the group-action of $\mathrm{e}^{t\bm{v}}$.

By the following Lemma\,\ref{lemma:ratio} and the mass conservation law $\widetilde{\mathcal{L}}_{\bm{v}} \varrho = 0$, we find
\begin{eqnarray}
\frac{\rmd}{\rmd t} {Q}
&=& \int_{\Sigma(t)} \widetilde{\mathcal{L}}_{\bm{v}}[ f(\omega^*/\varrho^*) \varrho]
\nonumber
\\
&=& \int_{\Sigma(t)} \left[ f' \varrho \widetilde{\mathcal{L}}_{\bm{v}} (\omega^*/\varrho^*) + f \widetilde{\mathcal{L}}_{\bm{v}}\varrho \right] = 0.
\label{enstrophy-2D}
\end{eqnarray}

\begin{lemma}
\label{lemma:ratio}
Let $\alpha$ and $\beta$ be a pair of  $n$-forms defined on a smooth manifold $M$ of dimension $n$.
Denoting $\alpha=\alpha^*\vol^n$ and $\beta=\beta^*\vol^n$,
%  ($\vol^n$ is the volume $n$-form),
we define $\vartheta = \alpha^*/\beta^*$. 
If $\widetilde{\mathcal{L}}_{\bm{V}} \alpha = 0$ and $\widetilde{\mathcal{L}}_{\bm{V}}\beta=0$ for a vector $\bm{V} \in TM$,
then $\widetilde{\mathcal{L}}_{\bm{V}} \vartheta = 0$.
\end{lemma}

\begin{proof}
By the definition, 
\[
\widetilde{\mathcal{L}}_{\bm{V}}\alpha = (\widetilde{\mathcal{L}}_{\bm{V}} \alpha^*)\vol^n + \alpha^* (\widetilde{\mathcal{L}}_{\bm{V}}\vol^n) = (\widetilde{\mathcal{L}}_{\bm{V}} \alpha^*)\vol^n + \alpha^* (\mathrm{div}\,\bm{V}) \vol^n.
\]
When $\widetilde{\mathcal{L}}_{\bm{V}}\alpha=0$, we may write $\widetilde{\mathcal{L}}_{\bm{V}} \alpha^* = -\alpha^* \mathrm{div}\,\bm{V} $.
The same formula applies to $\widetilde{\mathcal{L}}_{\bm{V}} \beta^*$. 
We thus have
\[
\widetilde{\mathcal{L}}_{\bm{V}} \left(\frac{\alpha^*}{\beta^*}\right)
=\frac{\widetilde{\mathcal{L}}_{\bm{V}}\alpha^*}{\beta^*} - \frac{\alpha^* \widetilde{\mathcal{L}}_{\bm{V}}\beta^*}{\beta^{*2}}
=  \frac{-\alpha^* \mathrm{div}\,\bm{V}}{\beta^*} + \frac{\alpha^* \beta^* \mathrm{div}\,\bm{V}}{\beta^{*2}} =0. 
\]
\qed
\end{proof}

We want to generalize ${Q}$ to a class of 3D systems by considering a  3-form integral of the form 
\begin{equation}
{Q} = \int_{V(t)} f(\vartheta) \varrho ,
\label{enstrophy-3}
\end{equation}
for  some scalar $\vartheta$ that reflects $\omega$.

%%%%%%%%%%%%%%%%%%%%%%%%%%%%%%%%%%%%%%%%%%%%%%%%%%%%%%%%%%%%%%%%%%%%%%%%%%%%%%%%%%%%%%%%%%%%%%%%
\section{Epi-two dimensional flow}
\label{sec:epi-2D}
\subsection{Reduction of the phase space}
\label{subsec:epi-2D}
A thought, drawn from the foregoing observation, is that the degeneration of 
one constant of motion (i.e. the helicity) must be compensated by a new constant of motion
(i.e. the generalized enstrophy).
Although we have seen that the degeneracy of the helicity is usual for 2D systems,
it may occur in a more general situation.
Then, it is conceivable that the compensation should also occur simultaneously.
If so, a generalized enstrophy may exist as a topological constraint in a wider class of ideal flows,
which we call \emph{epi-2D flows}.

\begin{definition}[epi-2D flow]
Let $Y$ be a phase space of smooth Clebsch parameters such that
\begin{equation}
Y = \{ \bm{\eta}=(\varphi,\varrho,q,p)^{\mathrm{T}} ;~\varphi, q \in \mathrm{\Omega}^0(T^3),~\varrho, p\in \mathrm{\Omega}^3(T^3) \} .
\label{epi-2D-1}
\end{equation}
The corresponding physical fields $\rho\Leftrightarrow \varrho^*$ and
\begin{equation}
\bm{V} \Leftrightarrow  \wp =  \rmd\varphi + \left(\frac{p^*}{\varrho^*}\right) \rmd q 
\label{epi-2D-2}
\end{equation}
are called epi-two-dimensional (epi-2D) flows.
\end{definition}

Notice the difference between $Y$ of (\ref{epi-2D-1}) and $Z$ of (\ref{2D_phase_space}); in particular,  the epi-2D flows are defined on the 3D domain $T^3$.
The reduced phase space $Y$ is a closed subset of $X$.
We denote by $\langle~,~\rangle_Y$ the reduced pairing of $Y^*$ and $Y$ (cf.\,(\ref{pairing})).
By restricting observables in $Y$, we define a canonical Poisson algebra $C^\infty_{\{~,~\}_Y}(Y)$,
where the Poisson bracket is
\begin{equation}
\{ F, G \}_Y = \langle \partial_{\bm{\eta}} F, J_Y \partial_{\bm{\eta}} G \rangle_Y,
\quad  
J_Y = J_c\oplus J_c .
\label{reduced_symplectic}
\end{equation}

% Another possibility of formulation is the use of the larger Poisson bracket of $X$ for reduced functionals such as
% $F(\bm{\eta})$ or $G(\bm{\eta})$.
% Such Poisson algebra, denoted by $C^\infty_{\{~,~\}}(Y)$, is \emph{noncanonical}.
% In fact, every observable such as $F(r,s)$ belong the center 

Epi-2D flow may have a finite vorticity $\rmd\wp=\rmd\check{p}\wedge\rmd q$,
where $\check{p}=p^*/\varrho^*$.
However, we observe
\[
\wp\wedge\rmd\wp = \rmd\varphi\wedge \rmd\check{p}\wedge\rmd q = \rmd(\varphi\wedge \rmd\check{p}\wedge\rmd q ),
\]
i.e. the helicity density $\wp\wedge\rmd\wp$ is an exact 3-form.  Hence, we have 

\begin{proposition}
\label{prop:zero-helicity}
Epi-2D flow has zero helicity, i.e. $C=\int_{{M}} \wp\wedge\rmd\wp =0$.
\end{proposition}

A \emph{vortex line} is a curve determined by
\begin{equation}
\frac{\rmd}{\rmd \tau} \bm{x} = \bm{\omega}(\bm{x}),
\label{vortex_line}
\end{equation}
where $\bm{\omega}$ is the vorticity.
For epi-2D flow,
$\bm{\omega}=\nabla\check{p}\times\nabla q~(\Leftrightarrow  \rmd\wp=\rmd\check{p}\wedge\rmd q)$.
Evidently, the level-sets of $\check{p}$ and $q$ are the `integral surfaces' of
vortex lines:
\begin{eqnarray*}
& &
\frac{\rmd}{\rmd \tau} \check{p}(\bm{x}(\tau)) = \nabla\check{p}\cdot\frac{\rmd}{\rmd \tau} \bm{x} =
\nabla\check{p}\cdot\bm{\omega}(\bm{x})=0,
\\
& &
\frac{\rmd}{\rmd \tau} q(\bm{x}(\tau)) = \nabla q\cdot\frac{\rmd}{\rmd \tau} \bm{x} \, =
\nabla q\cdot\bm{\omega}(\bm{x})=0.
\end{eqnarray*}
This well-known fact can be stated as

\begin{proposition}
\label{prop:integrable}
The vortex line equation of epi-2D flow is integrable;
two Clebsch parameters $\check{p}$ and $q$ define the integral surfaces.
We call the surface spanned by $\rmd\check{p}$ and $\rmd q$ the \emph{vortex surface}.
\end{proposition}

The epi-2D flow generated by a reduced Hamiltonian 
\begin{equation}
H(\bm{\eta}) = \int_{M} \left[ \frac12 {\big|\rmd\varphi + (p^*/\varrho^*) \rmd q \big|^2 }  + \varepsilon(\varrho^*)  \right]\,\varrho
\label{reduces-Hamiltonian}
\end{equation}
satisfies the 3D fluid equations (\ref{mass_conservation}) and (\ref{momentum}).

We may observe the epi-2D dynamics in the larger phase space $X$.  
Since the reduced Hamiltonian $H(\bm{\eta})$ does not include the variables $r$ and $s$,
the flow velocity $\bm{V}$ is independent of $r$ and $s$.
However, they obey the same equations (\ref{H-2}) and (\ref{H-3}), i.e.
\begin{equation}
\widetilde{\mathcal{L}}_{\bm{V}} r = 0,
\quad
\widetilde{\mathcal{L}}_{\bm{V}} s = 0.
\label{phantom}
\end{equation}
Such fields, co-moving with the epi-2D flow,
are called \emph{phantoms}\,\cite{YoshidaMorrisonFDR2014,YoshidaMorrison2016}.
Every functional such as $F(r,s)$ is a constant of motion:
$\{F(r,s), H(\bm{\eta}) \}=0$.

\subsection{Generalized enstrophy of epi-2D flow}
\label{subsec:g-enstrophy_epi-2D}
In light of the above, it is not surprising that we  have a family of conservation laws for epi-2D fluid motion:

\begin{theorem}
\label{theorem:enstropy_conservation}
Let $\bm{\eta}(t)$ be an epi-2D flow generated by the reduced Hamiltonian $H(\bm{\eta})$ of (\ref{reduces-Hamiltonian}),
and $r(t)$ be a co-moving phantom.
We define a generalized enstrophy
 (denoting $\omega=\rmd \check{p}\wedge \rmd q$)
\begin{equation}
{Q} = \int_{\Omega(t)} f(\vartheta)\, \varrho,
\quad \vartheta = \frac{(\omega \wedge \rmd r)^*}{\varrho^*}
% = \frac{(\rmd \check{p}\wedge \rmd q \wedge \rmd r)^*}{\varrho^*}
\label{enstrophy-2}
\end{equation}
with an arbitrary sooth function $f$ and an arbitrary co-moving 3D volume element $\Omega(t)\subset{M}$.
Then, $\rmd {Q}/\rmd t =0$.
\end{theorem}

\begin{proof}
Since $\widetilde{\mathcal{L}}_{\bm{V}}\varrho=0$, what we have to prove is $\widetilde{\mathcal{L}}_{\bm{V}} \vartheta =0$.
We have
\begin{eqnarray*}
\widetilde{\mathcal{L}}_{\bm{V}} (\rmd \check{p}\wedge \rmd q \wedge \rmd r)
&=& (\widetilde{\mathcal{L}}_{\bm{V}} \omega ) \wedge \rmd r +
\omega\wedge \widetilde{\mathcal{L}}_{\bm{V}} \rmd r
\\
&=& \omega\wedge \rmd (\widetilde{\mathcal{L}}_{\bm{V}} r) =0.
\end{eqnarray*}
By Lemma\,\ref{lemma:ratio},
we obtain $\widetilde{\mathcal{L}}_{\bm{V}} \vartheta =0$.
\qed
\end{proof}

The generalized enstrophy (\ref{enstrophy-2}) is a 3-dimensional generalization of the 
two-dimensional one (\ref{enstrophy-1}).
About its application, we have the following remarks:

\begin{remark}
\label{remark:appl1}
In a general 3D flow, ${Q}$ is also a constant of motion.
However, it does not characterize the vorticity, since $\omega$ must be inflated to
$\rmd\wp=\rmd \check{p}\wedge \rmd q + \rmd \check{s}\wedge \rmd r$.
\end{remark}

\begin{remark}
\label{remark:appl2}
In the case of 2D flow, we may first immerse the system into 3D by adding a perpendicular coordinate $z$ (see Sec.\,\ref{subsec:2D}),
and take the phantom $r=z$ (this $r$ is stationary).   
Then, Theorem\,\ref{theorem:enstropy_conservation} reproduces the
result of Sec.\,\ref{subsec:enstrophy}.
\end{remark}

\begin{remark}
\label{remark:appl3}
Suppose that $\omega\neq 0$.  Then, we may choose the initial value of $r$ so that
\begin{equation}
(\rmd \check{p}\wedge \rmd q \wedge \rmd r)^* = 
\frac{D(\check{p}, q, r)}{D(x,y,z)} \neq 0.
\label{jacobian}
\end{equation}
Hence, the generalized enstrophy can be made nontrivial.
Analogous to  the 2D generalized enstrophy, such an $r$ is a coordinate co-moving with the fluid,
which   penetrates the vortex surface (see Proposition\,\ref{prop:integrable}).
\end{remark}

%%%%%%%%%%%%%%%%%%%%%%%%%%%%%%%%%  added 4/28
\section{A particle picture}
\label{sec:particle}

%%%%%%%%%%%%%%%%%%%%%%%%%
\subsection{Epi-2D `particles'}
\label{subsec:particle}

We can exploit \emph{local}  epi-2D regions in order to define particle-like behavior. 
In the general 3D parameterization
$\bm{V} \Leftrightarrow  \wp = \rmd\varphi + \check{p}\rmd q + \check{s}\rmd r$,
a region in which $\check{s}=0$ may be called a  \emph{epi-2D domain}.
Since $\check{s}$ co-moves with the fluid,
every infinitesimal volume element (denoted  by $\Omega_j(t)$ with $j$ an index for each such volume element)
included in a \emph{epi-2D domain} may be viewed as a {quasiparticle},
which we call an \emph{epi-2D particle}.
The generalized enstrophy evaluated for $\Omega_j(t)$,
which we denote by ${Q}_1(\Omega_j)$ is a constant of motion, characterizing  the vorticity included there.
We call ${Q}_1(\Omega_j)$ the \emph{charge} of the {epi-2D particle} $\Omega_j$.

A symmetric epi-2D particle can be defined by a domain in which $\check{p}=0$, with the 
corresponding generalized enstrophy given by 
\[
{Q}_2(\Omega_j) = \int_{\Omega_j(t)} f(\vartheta_2)\,\varrho,
\quad
\vartheta_2 = \frac{(\omega_2\wedge\rmd q)^*}{\varrho^*},
\]
measuring the vorticity $\omega_2 = \rmd\check{s}\wedge\rmd r$.

As noted in Remark\,\ref{remark:appl1}, both
${Q}_1$ and ${Q}_2$ can be evaluated in a general co-moving domain (particle) $\Omega(t)\subset{M}$,
and they are ubiquitous constants.
However, they do not represent the `enstrophy' of an actual vorticity when the vorticity exists in a mixed state
$\rmd \check{p}\wedge\rmd q + \rmd\check{s}\wedge\rmd r$.
Hence, we may interpret ${Q}_1$ and ${Q}_2$ as  `potential' quantities, which
become `observable' when one of $\omega_1=\rmd \check{p}\wedge\rmd q$ or $\omega_2=\rmd\check{s}\wedge\rmd r$
alone occupies a domain. 

%%%%%%%%%%%%%%%%%%%%%%%%%
\subsection{Discovering epi-2D particle}
\label{subsec:discover}

In the preceding subsection, the notion of an epi-2D particle (or domain) was introduced using the Clebsch parameters which are the potential fields lying beneath the observables. Here we make an attempt to discover an epi-2D particle only from the  physical variable  $\bm{u}$.

We start by remembering the well-known relation:
\begin{lemma}[Frobenius]
\label{lemma:Frobenius}
Let $\wp$ be a $C^1$-class 1-form on a smooth manifold $M$ of dimension $n$ ($\geq 3$).
The following two conditions are equivalent:
\begin{enumerate}
\item
$\wp$ has \emph{zero helicity density}, i.e., 
\begin{equation}
\wp \wedge \rmd \wp = 0.
\label{zero-helicity-1}
\end{equation}

\item
$\wp$ is locally (i.e., in a neighborhood $\Omega$ of every point of $M$)
\emph{integrable}, i.e.,
there exist two scalars $\alpha$ and $\beta$ by which $\wp$ can be written as
\begin{equation}
\wp = \alpha \rmd \beta.
\label{zero-helicity-2}
\end{equation}
Then, the Pfaffian equation $\wp=0$ foliates $\Omega$ by the level-sets of $\beta$.
\end{enumerate}
\end{lemma}

% Let $V=\mathrm{\Omega}^1({M})$ be the space hosting the flow velocities. 
We define a quotient space
\begin{equation}
V_s = \mathrm{\Omega}^1({M}) / \rmd \mathrm{\Omega}^0({M}),
\label{solenoidal}
\end{equation}
which may be identified as the space of \emph{solenoidal vector fields}.
If we identify a 1-form $\wp\in V_s$ with a 3-vector field $\bm{V}$,
the integral
% we define 
% \begin{equation}
% \kappa = \wp\wedge\rmd\wp,
% \label{helicity_density}
% \end{equation}
% which corresponds to the \emph{helicity density}.
% In fact, identifying $\bm{V}\Leftrightarrow \wp$, 
$\int_{M} \wp\wedge\rmd\wp$ evaluates the helicity 
$C=\int_{M}\bm{V}\cdot\nabla\times\bm{V}\,\rmd^3x$.
Notice that the transformation $\bm{V} \mapsto\bm{V}+\nabla\phi$ ($\forall\phi$) does not change the helicity,
while the integrand (helicity density) $\bm{V}\cdot\nabla\times\bm{V}$ is modified.
But by defining the helicity density $ \wp\wedge\rmd\wp$ on $V_s$, the gauge ($\rmd \phi$) dependence has been removed.

If $\wp\wedge\rmd\wp= 0$ in $\Omega\subset{M}$, we say that $\wp$ is `helicity free' in $\Omega$.
By Lemma\,\ref{lemma:Frobenius},
a helicity-free $\wp$ can be represented as
$\wp = \check{p}\rmd q $ ($\exists \check{p},\,q$) in some $\Omega' \subset \Omega$,
which implies that $\wp' = \wp + \rmd \varphi $ ($\exists \varphi$) can be identified 
with the flow velocity $\bm{V}$ in $\Omega'$.
To put it in another way, we have

\begin{proposition}
\label{prop:integrability}
Given that the projection onto $V_s$ of a flow velocity $\bm{V}$ is helicity free in $\Omega\subset{M}$,
i.e.,  there exists $\nabla \phi$ by which we can make $\bm{V}-\nabla\phi \sim  \wp \in V_s$ such that 
$\wp\wedge\rmd\wp =0$ in $\Omega$,
then such a $\bm{V}$ is epi-2D in some $\Omega'\subset\Omega$.
\end{proposition}

%%%%%%%%%%%%%%%%%%%%%%%%%%%%%%%%
%%%%%%%%%%%%%%%%%%%%%%%%%%%%%%%%
%%%%%%%%%%%%%%%%%%%%%%%%%%%%%%%%
\section{Conclusion}
\label{sec:conclusion}

Diverse structures generated in fluids are not attributed to features of some  nontrivial energy functional.
In fact, the energy of a usual fluid is quite  simple, it being the equivalent a  \emph{norm} on  phase space 
of physical variables such as that given by (\ref{energy}). 
This is in marked contrast to the usual situation in condensed-matter physics where, 
for example, phase transitions or spinodal decompositions are modeled by bumpy energies.
The key role for fluids is, then, played by `constraints' that forbid the dynamics from obeying  simple orbits that might be determined by the energy alone. In the ideal (no-dissipation) limit, such constraints are manifested as conservation laws.
Indeed, ideal fluid mechanics has infinitely many such constants of motion, and some of them are essential for controlling  bifurcations of diverse structures or maintaining stability of some vortical motion.
In the present work, we focused on two well-known constants of motion:  the helicity of 3D flow and the (generalized) enstrophy of 2D flow, and we studied the basic mechanism of their creation.

In physics, a constant of motion is expected to be the product of some \emph{symmetry}.
However, the energy (Hamiltonian) of the fluid, represented in terms of the usual  physical  (Eulerian) variables, 
does not bear such symmetries to produce the helicity or enstrophy.
Therefore, we are led to consider a set of underlying basic parameters beneath the physical quantities,
and assume that some specific combinations of them appear as observables.
Here, we invoked Clebsch parameters, and showed that the helicity is the product of gauge symmetry of the Clebsch parameterization.  

We have observed that the phase space $X$ of general 3D flows is hierarchically foliated into submanifolds,  
where the smallest subsystem hosts vorticity-free (irrotational) flows. 
The next hierarchy $Y$ hosts the epi-2D flows, which is a subset of the zero-helicity leaf (Proposition\,\ref{prop:zero-helicity}). 
The subsystem $Y$ is foliated by the generalized enstrophy
(Theorem\,\ref{theorem:enstropy_conservation}),  which is a 3D generalization of the conventional one for  2D systems.
Notice that the reduction from $X$ (general 3D flow) to $Y$ (epi-2D flow) is not a geometrical constraint (cf.\ the 2D system of Sec.\,\ref{subsec:enstrophy}). 
However, epi-2D systems have  intrinsic vortex surfaces (Proposition\,\ref{prop:integrable}), 
which parallels the \emph{a priori} base space of the 2D system.
The generalized enstrophy is the measure of the circulation on such vortex surfaces.

\begin{acknowledgement}
ZY acknowledges discussions with Professor Y. Fukumoto and Y. Kimura.
The work of ZY was supported by JSPS KAKENHI Grant Number 23224014 and 15K13532, 
while that of PJM was supported by USDOE contract DE-FG02-04ER-54742. 
\end{acknowledgement}

%%%%%%%%%%%%%%%%%%%%%%%%%%%%%%%%%%%%%%%%%%%%%%%%%%%%%%%%%%%%%%%%%%%%%%%%%%%%%%%%%%%%%%%%%%%%%%
%%%%%%%%%%%%%%%%%%%%%%%%%%%%%%%%%%%%%%%%%%%%%%%%%%%%%%%%%%%%%%%%%%%%%%%%%%%%%%%%%%%%%%%%%%%%%%
%%%%%%%%%%%%%%%%%%%%%%%%%%%%%%%%%%%%%%%%%%%%%%%%%%%%%%%%%%%%%%%%%%%%%%%%%%%%%%%%%%%%%%%%%%%%%%
%%%%%%%%%%%%%%%%%%%%%%%%%%%%%%%%%%%%%%%%%%%%%%%%%%%%%%%%%%%%%%%%%%%%%%%%%%%%%%%%%%%%%%%%%%%%%%
%%%%%%%%%%%%%%%%%%%%%%%%%%%%%%%%%%%%%%%%%%%%%%%%%%%%%%%%%%%%%%%%%%%%%%%%%%%%%%%%%%%%%%%%%%%%%%
%%%%%%%%%%%%%%%%%%%%%%%%%%%%%%%%%%%%%%%%%%%%%%%%%%%%%%%%%%%%%%%%%%%%%%%%%%%%%%%%%%%%%%%%%%%%%%

% \input{referenc}

\end{document}